\documentclass[letterpaper]{article} 
\usepackage{aaai} 
\usepackage{times} 
\usepackage{helvet} 
\usepackage{courier} 
\usepackage{amsthm, amssymb, amsmath}
\usepackage[ruled,vlined]{algorithm2e}
\usepackage{graphicx}
\usepackage{booktabs}
\usepackage{color}
\usepackage{multirow}
\usepackage{url}

\usepackage{times}

\usepackage{tikz}
\usetikzlibrary{arrows,decorations,decorations.shapes,backgrounds,shapes}

\newtheorem{lemma}{Lemma}
\newtheorem{theorem}{Theorem}

\newcount\Comments
\definecolor{darkgreen}{rgb}{0,0.7,0}
\newcommand{\kibitz}[2]{\ifnum\Comments=1\textcolor{#1}{#2}\fi}

\newcommand{\cut}[1]{ }

\nocopyright

\title{Audit Games with Multiple Defender Resources} 
\author{Jeremiah Blocki$^1$, Nicolas Christin$^1$, Anupam Datta$^1$, Ariel D. Procaccia$^1$, Arunesh Sinha$^2$ \\ 
$^1$Carnegie Mellon University, USA; \{arielpro@cs., jblocki@cs., danupam@, nicolasc@\}cmu.edu \\
$^2$University of Southern California, USA; aruneshs@usc.edu }  %%%%%%%%%% % Body of Paper Begins 
\begin{document} 
\maketitle
\begin{abstract}
Modern organizations (e.g., hospitals, social networks, government agencies) rely heavily on audit to detect and punish 
insiders who inappropriately access and disclose confidential information. Recent work on \emph{audit games} models 
the strategic interaction between an auditor with a single audit resource and auditees as a Stackelberg game, augmenting 
associated well-studied security games with a configurable punishment parameter. We significantly generalize this audit 
game model to account for multiple audit resources where each resource is restricted to audit a subset of all 
potential violations, thus enabling application to practical auditing scenarios. We provide an FPTAS that computes 
an approximately optimal solution to the resulting non-convex optimization problem. The main technical novelty is in the design and correctness 
proof of an optimization transformation that enables the construction of this FPTAS. In addition, we experimentally 
demonstrate that this transformation significantly speeds up computation of solutions for a class of audit games and 
security games. 
\end{abstract}

\section{Introduction}

Modern organizations (e.g., hospitals, banks, social networks, search
engines) that hold large volumes of personal information rely heavily on
auditing for privacy protection. These audit mechanisms
combine automated methods with human input to detect and punish violators.
Since human audit resources are limited, and often not sufficient
to investigate all potential violations, current state-of-the-art audit tools
provide heuristics to guide human effort~\cite{fairwarning}. However, numerous 
reports of privacy breaches caused by malicious insiders bring to question the
effectiveness of these audit mechanisms~\cite{Ponemon2011benchmark,Ponemon2011}. 

Recent work on \emph{audit games} by Blocki et al.~\shortcite{bcdps2013audit} 
approaches a piece of this problem using game-theoretic techniques. 
Their thesis is that effective human audit resource allocation and punishment
levels can be efficiently computed by modeling the audit process as a
game between an auditor and auditees. At a technical level, their audit 
game model augments a well-studied Stackelberg security games model~\cite{tambe}
with a configurable punishment parameter. The auditor (henceforth called 
the \emph{defender}) can audit \emph{one} of $n$ potential violations (referred to as 
\emph{targets}). The defender's optimal strategy is a randomized auditing 
policy---a distribution over targets such that when the attacker best responds 
the defender's utility is maximized. The novel ingredient of the audit games model 
is a punishment level, chosen by the defender, which specifies how severely the 
adversary will be punished if he is caught. The defender may try to set a high 
punishment level in order to deter the adversary. However, punishment is not
free. The defender incurs a cost for punishing, e.g. punishments such as suspension
or firing of violators require maintaining resources for hiring and training of 
replacements. Blocki et al.~\shortcite{bcdps2013audit} provide an efficient algorithm 
for computing an optimal strategy for the defender to commit to. 

While this work distills the essence of the defender's dilemma in auditing situations, 
it is too restricted to inform real-world audit strategies. In typical audit 
settings, the defender has \emph{multiple resources} using which she can audit a subset of the 
targets (not just one target). Furthermore, each resource may be \emph{restricted} in 
the targets that it can audit. For example, some organizations follow 
hierarchical audit strategies in which a manager is only required to audit potential 
violations committed by her direct reports. Similarly, specialized audits are 
common, for example, employing disjoint audit resources to detect finance-related and privacy-related 
violations.

\smallskip
\noindent\textbf{Our Contributions.}
%In this paper, we generalize the simplified audit model of \cite{bcdps2013audit} to include elements that would occur in a realistic audit setting. 
%Our most important contribution is dealing with multiple audit resources, where each audit resource may be constrained with respect to the targets it can audit. We view this extension as absolutely crucial, as almost any organization conducting audits will employ multiple auditors, who can carry out multiple audits each. Examples of constrained auditing
% include billing-related cases that must be investigated specifically by financial auditors, and localized auditing in which managers audit their direct subordinates.
We present a generalized Stackelberg audit game model that accounts for multiple audit resources where 
each resource is restricted to audit a subset of all potential violations, thus enabling application to the 
practical auditing scenarios described above. Our main theoretical result is a \emph{fully polynomial time 
approximation scheme (FPTAS)} to compute an approximate solution to the resulting non-convex 
optimization problem. 

To arrive at this FPTAS, we begin with a simple \emph{fixed parameter tractable (FPT) algorithm} that 
reduces the non-convex optimization problem to a linear programming problem by fixing the 
punishment parameter at a discrete value. Since we perform a linear search over all possible 
discrete values of the punishment parameter over a fixed range in increasing intervals of size 
$\epsilon$, we get an FPT algorithm once the bit precision is fixed.  

Next we present an \emph{optimization transformation} that reduces the number of variables in the 
optimization problem at the cost of generating additional constraints. We also provide sufficient 
conditions that guarantee that the number of constraints is polynomial in size. Significantly, these 
conditions are satisfied in important practical auditing scenarios, including the hierarchical and 
specialized audit settings discussed earlier. The design and correctness proof of this optimization 
transformation constitutes the central novel technical contribution of this paper. 
Finally, we present an FPTAS to compute the defender's strategy leveraging the output of the 
optimization transformation when it generates a polynomial number of constraints.

In addition to its role in enabling the design of an FPTAS, a practical motivation for designing the 
optimization transformation is its promise of speeding up computing game solutions using the FPT 
algorithm. We experimentally demonstrate that the transformation produces speedups of up to $3\times$
for audit game instances and over $100\times$ for associated security game instances. In general, 
the speedups are higher as the problem size increases. 
%Our experiments handled larger instances of security games since they run faster. 

As an additional contribution, we consider audit scenarios where the defender can set a different 
punishment level for each target instead of setting a single, universal punishment level. We provide an 
FPT algorithm for this problem by designing a novel reduction of this setting to a second order cone 
program (SOCP).

\smallskip
 \noindent\textbf{Related Work.}
% As stated earlier, our work improves upon the simple model from an earlier paper~\cite{bcdps2013audit}
%and provides a realistic model of auditing. We extend their approach to obtain a FPTAS in our model, but with certain restrictions. We first require to
%transform our optimization problem into another equivalent problem with fewer variables, but more constraints. We obtain a FPTAS when the set of new
%constraints obtained by the transformation is polynomial in size. 
Our work is most closely related to the paper of Blocki et al~\shortcite{bcdps2013audit}. We elaborate on the technical connections in the exposition and analysis of our results. 

Our work is also closely related to work on security 
games~\cite{tambe}. The basic approach to solving security games, and Stackelberg games more generally (where a leader chooses an optimal strategy, i.e., one that maximizes its utility assuming that the follower best responds) was introduced by Conitzer and Sandholm~\shortcite{ConitzerS06}; it does not scale well when the number of defender strategies is exponentially large, as is the case in most security games. Typical algorithms rely on formulating a mixed integer linear program, and employing heuristics; these algorithms do not provide provable running-time guarantees (see, e.g., Paruchuri et al.~\shortcite{paruchuri2009coordinating}; Kiekintveld et al.~\shortcite{Kiekintveld1558108}).
Often the problem is solved just for the coverage probabilities or marginals (probability of defending a 
target) with the hope that it would be implementable (i.e., decomposable into a valid distribution over allocations). In contrast, Korzhyk et al.~\shortcite{KorzhykCP10} give a polynomial time algorithm for security games where the attacker has multiple resources, each of which can only protect one target at a time (this restriction is known as \emph{singleton schedules}). Their main tool is the Birkhoff-Von Neumann Theorem, which we also apply. 

All our algorithms have theoretical results about their efficiency. Our optimization
transformation transforms the problem to one with only coverage probabilities as variables and adds additional constraints for the coverage probabilities that restricts
feasible coverage probabilities to be implementable. Thus, in contrast with Korzhyk et al.~\shortcite{KorzhykCP10}
our problem has much fewer variables at the cost of additional constraints.
Also, the punishment parameter makes our problem non-convex, and so our algorithms must leverage a suite of additional techniques.

\section{The Audit Games Model} \label{model}
An audit game features two players: the defender ($D$), and the        
adversary ($A$). The defender wants to audit $n$ targets $t_1, \ldots,  
t_n$, but has limited resources which allow for auditing only some      
of the $n$~targets. 
%The defender employs human auditors, with each      
%human auditor performing a fixed number of inspections. 
%Viewing each    
%inspection of an auditor as a resource, 
Concretely, these resources 
could include time spent by human auditors, computational resources
used for audit (e.g., automated log analysis), hardware (e.g., cameras) 
placed in strategic locations, among many other examples. The exact     
nature of these resources will depend on the specific audit problem     
considered. Rather than focusing on a specific type of audit, we denote 
resources available to the defender for audit as {\em inspection        
resources}. The                                                         
defender has $k$~{\em inspection 
resources} $\{s_1, \ldots, s_k\}$ at her disposal, with $k<n$.
Each inspection resource can be used to audit at most one target. 
Inspection resources are further constrained by the set of targets    
that they can audit: for instance, a human auditor may not have the     
expertise or authority 
to audit certain targets. We define a set of tuples 
$R$ such that a tuple $(j,i) \in R$ indicates that inspection resource  
$s_j$ cannot be used to audit target $t_i$.                             

A pure action of the defender chooses the allocation of inspection resources to 
targets. %Such an allocation can be represented by a matrix $A=(a_i^j)$ where $a_{i}^j$ is 
%$1$ (rows are resources, columns targets) if resource $s_j$ is allocated to target $t_i$, otherwise it is $0$. 
A randomized strategy is given by a probability distribution over pure actions. 
%The natural way to represent a randomized strategy 
%is to enumerate the set of pure actions that the defender can play along with the probability the defender plays each pure action, but
%this representation may require exponential space. Instead we represent a randomized strategy by a matrix of probabilities, with
We use a compact form to represent a randomized strategy as a matrix    
of probabilities, with $p_{i}^j$ the probability of inspection resource $s_j$    
auditing target $t_i$ subject to the following constraints              
$$
\renewcommand{\arraystretch}{1.0}
\begin{array}{l l}
p_i = \sum_{j=1}^{k} p_{i}^j \leq 1, \quad \sum_{i=1}^{n} p_{i}^j \leq 1 \quad \mbox{for all $i,j$ and } &\\
p_{i}^j = 0 \mbox{ for all $(j, i) \in R$ and } \forall (j, i). \ p_{i}^j \geq  0 \mbox{ \ ,} &
\end{array}
\renewcommand{\arraystretch}{1.0}
$$
where $p_i$ denotes the probability that target $t_i$ is inspected. 
These constraints can be represented as {\em grid constraints} as
follows:
\renewcommand{\arraystretch}{1.1}
\begin{center}
\begin{tabular}{c c c c c c }
&$t_1$ & \multicolumn{2}{c}{$\ldots $} & $t_n$\\
\cline{2 -5}
$s_1$ & \multicolumn{1}{|c}{$p_1^1$} & \multicolumn{1}{|c}{... \ }  &
\multicolumn{1}{|c}{...} & 
\multicolumn{1}{|c|}{$p_n^1$} & $\sum_i p^1_i \leq 1$\\
\cline{2 -5}
... & \multicolumn{1}{|c}{...} & \multicolumn{1}{|c}{...} & 
\multicolumn{1}{|c}{...} &
\multicolumn{1}{|c|}{...} & ... \\
\cline{2 -5}
... & \multicolumn{1}{|c}{...} & \multicolumn{1}{|c}{...} & 
\multicolumn{1}{|c}{...} &
\multicolumn{1}{|c|}{...} & ... \\
\cline{2 -5}
$s_k$ & \multicolumn{1}{|c}{$p_1^k$} & \multicolumn{1}{|c}{...} & 
\multicolumn{1}{|c}{...} &
\multicolumn{1}{|c|}{$p_n^k$} & $\sum_i p^k_i \leq 1$\\
\cline{2 -5}
 & $\sum_j p^j_1 \leq 1$ & \multicolumn{2}{c}{$\ldots$} &
$\sum_j p^j_n \leq 1$ &\\

\end{tabular}
\end{center}
\renewcommand{\arraystretch}{1.0}
%This efficient representation of a randomized 
%strategy is justified by the fact that
Such a matrix can be decomposed into pure actions efficiently. (See      
the full version for the Birkhoff-von Neumann~\shortcite{MR0020547} result presented in           
Korzhyk et al.~\shortcite{KorzhykCP10}, which enables the decomposition.) Furthermore, for  
every distribution over pure actions we can define an ``equivalent''    
strategy using our compact representation: the distributions will be  
equivalent in the sense that for every target $t_i$ the probability     
$p_i$ that that target is inspected is the same.                       

Similarly to the basic audit games model of Blocki et                   
al.~\shortcite{bcdps2013audit}, the defender also chooses a punishment  
``rate'' $x \in [0,1]$ such that if auditing detects an attack (i.e.,   
violation), the attacker is fined an amount~$x$. The adversary attacks  
one target such that given the defender's strategy the adversary's      
choice of attack is the best response.

%It is possible that adversary does not attack any target (no violation); to allow such a possibility
%we include a dummy target for which all associated costs are zero. 
%The dummy target forces the resource to be allocated to 
%audit some target, i.e., zero probability of not auditing.

We follow the notations of earlier papers on security games: 
let the utility of the defender be $U^a_D (t_i)$ when audited target $t_i$ was 
found to be attacked, 
and $U^u_D (t_i)$ when unaudited target $t_i$ was found to be attacked.
We assume that the attacks on unaudited targets are discovered by 
an external source.
Similarly, define the 
utility of the attacker as $U^a_A (t_i)$ when the attacked target $t_i$ is audited, 
and $U^u_A (t_i)$ when attacked target $t_i$ is not audited, excluding
any punishment imposed by the defender.
Attacks discovered externally are costly for the defender, thus, $U^a_D (t_i) \geq U^u_D (t_i)$.
Similarly, attacks not
discovered by internal audits are more beneficial to the attacker, 
thus, $U^u_A (t_i) \geq U^a_A (t_i)$. 

%Following the naming
%convention of Korzhyk, et al.~\cite{KorzhykCP10}, we have a game model with heterogeneous resources and singleton schedules, but with a critical difference.
%The additional component in audit games is punishments. The defender chooses
%a punishment ``rate'' $x \in [0,1]$ such that if auditing detects an
%attack, the attacker is fined an amount~$x$. The defender pays a cost for setting 
%punishment rate $x$, which we
%model as a linear function $ax$, where $a > 0$; however, our
%results directly extend to any cost function polynomial in $x$. Assuming
%$x \in [0,1]$ is also without loss of generality as utilities can be
%scaled to be comparable to $x$. We do assume the punishment rate is
%fixed and deterministic; a natural assumption that corresponds to a
%consistent policy.

We can now define the full utility functions. The probabilities 
$p_1, \ldots, p_n$ of each target being audited are given by 
$p_i = \sum_{j=1}^{k} p_{i}^j$. Then, the utility of the 
defender when target $t_*$ is attacked is 
$
p_{*}U^a_D (t_*) + (1 - p_{*})U^u_D (t_*) - ax.
$
The defender pays a fixed cost $ax$ regardless of the outcome, where $a$ is a constant. 
The losses captured by the $-ax$ term includes loss due to creation of a fearful work environment and cost incurred in maintaining a surplus of employees in anticipation of suspension.
In the same scenario, the 
utility of the 
attacker when target $t_*$ is attacked is 
$
p_{*}(U^a_A (t_*) - x) + (1 - p_{*})U^u_A (t_*).
$
The attacker suffers punishment $x$ only when attacking an audited target. 

A possible extension of the model above is to account for immediate losses that the defender suffers by imposing a punishment, e.g., firing or suspending an employee requires time and effort to find a replacement. Mathematically, we can account for such losses by including an additional term within the scope of $p_*$  in the payoff of the defender:
$
p_{*}(U^a_D (t_*) - a_1 x ) + (1 - p_{*})U^u_D (t_*) - ax 
$, 
where $a_1$ is a constant. All our results (FPT and FPTAS algorithms) can be readily extended to handle
this model extension, which we present in the full version.

\medskip
\noindent\textbf{Equilibrium.} 
Under the Stackelberg equilibrium solution, the defender commits to a (randomized) strategy, followed by a best response by the adversary; the defender's strategy should maximize her utility.
The mathematical problem involves solving multiple optimization 
problems, one each for the case when attacking $t_*$ is in fact the best response of 
the adversary. Thus, assuming $t_{*}$ is the best response of the adversary, the 
$*^{th}$ optimization problem $P_*$ in our audit games setting is
$$
\begin{array}{ll}
\displaystyle\max_{p_{ij},x} & p_{*}U^a_D (t_*) + (1 - p_{*})U^u_D (t_*) - ax \ ,\\
\mbox{subject to}& \forall i \neq *. \ p_{i}(U^a_A (t_i) -x) + (1 - p_{i})U^u_A (t_i)  \\
%&  \forall i \neq *. \ p_{i}(U^a_A (t_i) -x) + (1 - p_{i})U^u_A (t_i)  &\\
&  \ \ \ \ \ \ \ \
\leq p_{*}(U^a_A (t_*) - x) + (1 - p_{*})U^u_A (t_*) \ , \\
& \forall j. \ 0 \leq \sum_{i=1}^{n} p_{i}^j \leq 1 \ ,  \\
& \forall i. \ 0 \leq p_i = \sum_{j=1}^{k} p_{i}^j \leq 1 \ , \forall (j, i). \ p_{i}^j \geq  0 \ ,\\
& \forall (j,i) \in R. \  p_{i}^j = 0 \ , \ \ \ 0 \leq x \leq 1 \ .
\end{array}
$$
The first constraint verifies that attacking $t_*$ is indeed a best response for the adversary. 

The auditor solves the $n$ problems $P_1, \ldots, P_n$ (which correspond to the cases where the best response is $t_1,\ldots,t_n$, respectively), and chooses the best among 
all these solutions to obtain the final strategy to be used for auditing. This is a generalization of the multiple LPs approach of Conitzer and Sandholm~\shortcite{ConitzerS06}.

For ease of notation, let 
$\Delta_{D,i} = U^a_D (t_i)  - U^u_D (t_i)$, $\Delta_{i} = 
U^u_A (t_i) - U^a_A (t_i)$ and $\delta_{i,j} = U^u_A (t_i) - U^u_A (t_j)$.
Then, $\Delta_{D, i} \geq 0$, $\Delta_{i} \geq 0$, and
 the objective can be written as $p_n \Delta_{D, *} - ax$, subject to the quadratic constraint $$
 p_i (-x - \Delta_i) + p_n (x + \Delta_{*}) + \delta_{i,*}  \leq 0 \ .$$

Without loss of generality we will focus on the $n$'th program $P_n$, that is, we let $*$ be $n$.

\smallskip
\noindent\textbf{Inputs.} The inputs to the above problem are specified in $K$ bit 
precision. Thus, the total length of all inputs is $O(nK)$. 

%\medskip
%\noindent\textbf{Multiple Attacks.} A situation where the adversary attacks multiple targets can be intuitively
%modeled as a simple extension, which we state informally. Suppose the attacker attacks $m$ targets. The objective is then a linear function of corresponding probabilities of 
%attacking the $m$ targets and $x$. Then, add additional constraints to enforce
%the condition that the utility 
%of the attacker in attacking any of the $m$ targets is higher than the utility
%in attacking any of the other $n-m$ targets. The Stackelberg equilibrium can be computed by solving $n \choose m$ such optimization problems. Also, the case of at-most $m$ attacks will require
%solving $\sum_{i=1}^{m} $$n \choose i$ optimization problems, which is polynomially many for constant $m$.

\section{Fixed-Parameter Tractable Algorithms} \label{algorithm}
In this section, we present our FPT algorithm for optimization problem $P_n$, followed by the optimization transformation that improves the FPT algorithm and enables the FPTAS in the next section. Finally, we briefly describe an extension of our algorithmic results to target-specific punishments. 

We start with the FPT for $P_n$. Our algorithm is based on the following straightforward observation: if we fix the punishment level $x$ then the $P_n$ becomes a linear program that can be solved in polynomial time. We therefore solve the optimization problem $P_n$ for discrete values of $x$ (with interval size $\epsilon$) and take the solution that maximizes the defender's utility. This approach provides the following guarantee:
%\ap{In the text we need to be more specific about what ``good approximation'' means. In the statement of the lemma, approximation doesn't appear. An FPT algorithm is, by default, optimal, which is clearly not the case here.}  
%To address the special case when $x^o$ is very small we show that second order cone programming can be used to find a solution that approximates the optimal solution. Our approach yields an FPT algorithm for each problem  $P_i$: the running time is exponential in the bit precision used to specify the input, and polynomial in the other parameters. 
%\jb{I believe we can claim that the algorithm is FPT when the bit precision $K$ is logarithmic in the input size.}
\begin{theorem} \label{FPTLemma}
The above approach of solving for discrete values of $x$ is a FPT $\Theta(\epsilon)$-additive approximation algorithm for the problem $P_n$ if either the optimal value of $x$ is greater than a small constant or $\Delta_n \neq 0$; the bit precision is the fixed parameter.
\end{theorem}
\begin{proof}[Proof Sketch]
The proof proceeds by arguing how much the objective changes when the value of $x$ is changed by less
than $\epsilon$. The exact algebraic steps are in the full version.
\end{proof}
We emphasize that fixing the bit precision is reasonable, because inputs to the game model are never known with certainty, and therefore high-precision inputs are not used in practice (see, e.g., Nguyen et al.~\shortcite{Nguyen:2014:SCU:2615731.2615784}, Kiekintveld et al.~\shortcite{Kiekintveld:2013:SGI:2484920.2484959}, Blum et al.~\shortcite{BHP14}).

%\section{Fully Polynomial Time Approximation Scheme}

A na\"ive approach to improving our FPT algorithm is to conduct a binary search on the punishment rate $x$. This approach may fail, though, as the solution quality is not single-peaked in $x$. We demonstrate this in the full version using an explicit example.
%The optimization problem for any fixed value of 
%$x$ is a linear programming problem. Obtaining an additive approximation of $\epsilon$
%requires discretization intervals of size $\theta(\epsilon)$ (See Appendix~\ref{appendix:FPT}). Thus, the 
%running time is $O(\mathsf{LP}(n)/\epsilon)$, where $\mathsf{LP}(n)$ is the time to 
%solve a linear program with $n$ variables. 
%This FPT approach can also deal with scenario in which attackers attacks multiple
%targets bounded by a constant. Suppose the attacker attacks $m$ targets. Discretizing
%$x$ would yield linear programs for each fixed value of $x$. As $m$ is a constant the 
%size of these linear programs is within a constant factor of the single attack FPT.
%In an earlier paper~\cite{bcdps2013audit}, an FPTAS algorithm was proposed to compute equilibrium in  
%audit games for the
%simple case of one inspection.  We show that for many realistic audit scenarios the number of constraints produced by the transformation is polynomial in size.
Instead, we describe a transformation of the optimization problem, which will enable
a FPTAS for our problem under certain restrictions.
%An important practical consequence of the
%transformation is the reduction of running time of the FPT algorithm, as a result of the reduction in the number of variables of the problem. We describe this transformation followed by the conditions under
%which polynomially many new constraints are produced.

\subsection{Extracting constraints for $p_i$'s}
The transformation 
eliminates variables $p_i^j$'s and instead extracts inequalities (constraints) for the variables
$p_i$'s from the constraints below (referred to as grid constraints)
$$
\renewcommand{\arraystretch}{1.0}
\begin{array}{l}
\forall i. \ 0 \leq p_i = \sum_{j=1}^k p_i^j \leq 1 \ , \forall j. \ 0 \leq \sum_{i=1}^n p_i^j \leq 1 \ ,  \\
\ \forall (j,i). \ p_i^j \geq  0 \ , \forall (j,i) \in R. \  p_{i}^j = 0 \ . 
\end{array}
$$

Consider any subset of inspection resources $L$ with the resources in $L$ indexed by $s_1, \ldots, s_{|L|}$ ($|L| \leq k$). We let $M=OnlyAuditedBy(L) \subset \{t_1$,$\ldots,t_n\}$ denote the subset of targets that can only be audited by a resource in $L$ (e.g., for every target $t_i \in M$ and every resource $s_j \notin L$ the resource $s_j$ cannot inspect the target $t_i$) . For notational convenience we assume that $M$ is indexed by 
$t_1, \ldots, t_{|M|}$. Then,  in case $|L| < |M|$, we obtain a constraint $p_{t_1}+ 
\ldots + p_{t_{|M |}}\leq |L|$ because there are only $|L|$ resources that could be used to audit these $|M|$ targets. We call such a constraint $c_{M,L}$. Consider the set of
all such constraints $C$ defined as
$$
\{c_{M,L} ~|~ L \in 2^S, \mbox{$M = OnlyAuditedBy(L)$}, |L| < |M|\}
$$
where $S= \{s_1, \ldots, s_k\}$ and $T = \{t_1, \ldots, t_n\}$.

\begin{lemma} \label{reduction}
The optimization problem $P_*$ is equivalent to the optimization problem obtained by replacing
the grid constraints by $C \cup \{0 \leq p_i \leq 1\}$ in $P_*$.
\end{lemma}
\begin{proof}[Proof Sketch]
We present a sketch of the proof with the details in the full version.
As the optimization objective depends on variables $p_i$'s only, and quadratic constraints are identical in both problems we just need to show that 
the regions spanned by the variables $p_i$'s, as specified by the linear constraints, are the same in both problems. As one direction of the inclusion is easy, we show the harder case below.

Let $C^+$ denote the convex polytope $C \cup \{0\leq p_i\leq 1 \}$. Given a
point $(p_1,...,p_n) \in C^+$ we want to argue that we can find values
$p_i^j$'s satisfying all of the grid constraints. We first note that it
suffices to argue that we can find feasible $p_i^j$'s for any extreme point
in $C^+$ because any point in $C^+$ can be written as a convex combination of
its extreme points (Gallier 2008). Thus, we could find feasible $p_i^j$'s
for any point in $C^+$ using this convex combination.

In the full version we prove that each extreme point in $C^+$ sets the
variables $p_1,...,p_n$ to $0$ or $1$. Let $k'$ denote the number of ones in an
extreme point. Note that $k' \leq k$ because one of the inequalities is $p_1 +
... + p_n \leq k$.
Consider the undirected bipartite graph linking the inspection nodes to the target nodes, 
with a link indicating that the inspection can audit the linked target. This graph is known
from our knowledge of $R$, and each link in the graph can be labeled by one of the $p_i^j$
variables. Let $S'$ be the set of targets picked by the ones in any
extreme points. We claim that there is a perfect matching from $S'$ to the the set of 
inspection resources (which we prove in next paragraph). Given such a perfect matching, assigning
$p_i^j = 1$ for every edge in the matching yields a feasible solution, which completes
the proof.

We prove the claim about perfect matching by
contradiction. Assume there is no perfect matching, then there must be a set 
$S'' \subseteq S'$, such that $|N(S'')| < |S''|$ ($N$ is the neighbors function --- this result follows from Hall's theorem). As $S ''\subseteq S'$ it must hold that 
$p_i = 1$ for all $i \in index(S'')$ (function index gives the indices of the set of targets). Also, the set of targets $S''$ is audited
only by inspection resources in $N(S'')$ and, by definition of $C$, we must have a constraint
$
\sum_{i \in index(S'')} p_i \leq |N(S'')| \ .
$
Using $|N(S'')| < |S''|$, we get 
$
\sum_{i \in index(S'')} p_i < |S''| \ .
$
But, since $|index(S'')| = |S''|$, we conclude that all $p_i$ for targets in $S''$ cannot be 
one, which is a contradiction.
\end{proof}
 Observe that obtaining $p^j_i$'s from the $p_i$'s
involves solving a linear feasibility problem, which can be done efficiently.

Importantly, the definition of $C$ is constructive and provides an algorithm to compute it. However, the algorithm has a worst-case running time exponential in $k$. 
Indeed, consider $k$ resources $s_1, \ldots, s_k$ and $2k$ targets $t_1, \ldots, t_{2k}$. Each resource $s_i$ can inspect targets
$t_1, t_2, t_{2i-1} , t_{2i}$. For each set of $k/2$ resources $L \subseteq \{s_2$,$\ldots,s_k\}$ we have $M =$ $OnlyAuditedBy(L)$ $=\{t_1,t_2\}\cup \left(\bigcup_{s_j \in L} \{t_{2j-1},t_{2j} \} \right)$ . Observe that $|M| = k+2 > k/2 = |L|$ so for each $L \subseteq \{s_2$,$\ldots,s_k\}$ we get a new constraint $c_{M,L}$. Thus, we get $k-1 \choose k/2$
constraints.

\subsection{Conditions for Poly. Number of Constraints} \label{polyconstraints}
Motivated by the above observation, we wish to explore an alternative method for computing $C$. We will also provide sufficient
conditions under which $|C|$ is polynomial. 

%Before we present our algorithm we first present the following counterexample, which demonstrates that in the worst case the number of constraints in $C$ will be super-polynomial.
 
%\jb{The old example where $s_i$ could inspect targets $t_1,...,t_{2i}$ did not make %sense to me so I replaced it with the above example. Arunesh: please double check %this example.}
%  Consider any set of $k$ targets formed by picking $k/2$ target pairs among
%$t_{2i-1}, t_{2i}$ for any $i$, except $i=1$. The number of such sets are $k-1 \choose k/2$.
%Each such set of targets results in a valid constraint, since the number of resources 
%that can audit these targets in less than $k$. This is because the first resource only
%inspects $t_1, t_2$, and that pair is never chosen. Thus, we get $k-1 \choose k/2$
%constraints, which is  not polynomial in $k$.

The intuition behind our alternative algorithm is that instead of iterating over sets of inspection resources, we could iterate over sets of targets. As a first step, we identify equivalent targets and merge them. Intuitively, targets that can be audited by the exact same set of inspections are equivalent. Formally, $t_i$ and $t_k$ are equivalent if $F(t_i) = F(t_k)$ where  $F(t_\ell) = \{s_j~\vline~(j,\ell)\notin R\}$. 

The algorithm
is formally given as Algorithm~\ref{algmult2}.
%In the above algorithm we treat paths as walks with no repeating edges. Note that we do 
%not use simple paths. 
It builds an intersection graph from the merged targets: every merged set of targets is a node, and two 
nodes are linked if the two sets of inspection resources corresponding to the 
nodes intersect. The algorithm iterates through every connected induced
sub-graph and builds constraints from the targets associated with the 
nodes in the sub-graphs and the set of inspection resources associated with them.
The next lemma proves the correctness of the Algorithm~\ref{algmult2}.
%Next, we prove the  After that, 
%we investigate the conditions under which the Algorithm~\ref{algmult2} runs efficiently, and yields a polynomial sized $C$.

\begin{algorithm}[ht] \label{algmult2}
\DontPrintSemicolon
Compute $F$, the map from $T$ to $2^{\{s_1, \dots, s_k\}}$ using $R$.\;
Merge targets with same $F(t)$ to get set $T'$ and a map $W$, where $W(t') = $ \#merged targets that yielded $t'$\;
Let $PV(t')$ be the set of prob. variables associated with $t'$, one each from 
the merged targets that yielded $t'$\;
Form an intersection graph $G$ with nodes $t' \in T'$ and edge set $E = \left\{\{t_i\rq{},t_k\rq{}\} ~\vline F(t_i\rq{})\cap F(t_k\rq{})\neq \emptyset \right\}$ 

%and each node denoting set $F(T')$\;

$L \leftarrow \mathsf{CONNECTEDSUBGRAPHS}(G)$ \;
$C \leftarrow \phi$\;
\For{$l \in L$}{
	Let $V$ be all the vertices in $l$ \;
	$P \leftarrow \bigcup_{v \in V} PV(v)$\;
	$k \leftarrow \sum_{v \in V} W(t')$\;
	\If{$|P| > k$}{
		$C \leftarrow C \cup \{\sum_{p \in P} p \leq k\}$
	}
}
\Return $C$\;
\caption{$\mathsf{CONSTRAINT\_FIND}(T, R)$}
\end{algorithm}

\begin{lemma} \label{samealg}
$\mathsf{CONSTRAINT\_FIND}$ outputs constraints that define the same 
convex polytope in $p_1, \ldots, p_n$ as the constraints output by the na\"ive
algorithm (iterating over all subsets of resources).
\end{lemma}
The proof appears in the full version. The algorithm  is clearly not polynomial time in general, because it iterates over all connected subgraphs. The next lemma provides sufficient conditions for polynomial running time.
%We state sufficient conditions
%for $\mathsf{CONSTRAINT\_FIND}$ to %run in poly time. 

\begin{lemma} \label{graphlemma}
$\mathsf{CONSTRAINT\_FIND}$ runs in polynomial time if at least one of the following conditions holds:
\begin{itemize}
\item The intersection graph has $O(\log n)$ nodes.
\item The intersection graph has constant maximum degree and a constant number of nodes with degree at least $3$.
\end{itemize}
\end{lemma}
\begin{proof}[Proof Sketch]
The detailed proof is in the full version. It is not hard to observe that we need sufficient conditions
for any graph to have polynomially many induced connected sub-graphs. The first case above is obvious
as the number of induced connected sub-graphs in the worst case (fully connected graph) is $2^N$, where
$N$ is number of nodes. The second case can be proved by an induction on the number of nodes in 
the graphs with degree greater than $3$. Removing any such vertex results in a constant number of disconnected components (due to constant max degree). Then, we can argue that the number of connected
sub-graphs of the given graph will scale polynomially with the max number of connected sub-graphs
of each component. The base case involves graphs of degree less than two, which is a graph with paths and cycles and such a graph has polynomially many connected sub-graphs.
\end{proof}

\paragraph{Why Are These Conditions Realistic?}
The conditions specified in Lemma \ref{graphlemma} capture a wide range of practical audit scenarios.

First, many similar targets can often be grouped together by type. In   
a hospital case, for instance, rather than considering each individual  
health record as a unique target worthy of specific audit strategies,   
it might make more sense to have identical audit policies for a small   
set of patient types (e.g., celebrities, regular folks...). Likewise,   
in the context of tax audits, one could envisage that individuals are
pooled according to their types (e.g., high income earners, expatriates, ...). In      
practice, we expect to see only a few different types.                  
%Types of targets could include ``celebrity health records,'' accesses or financial cases. 
%The number of types of targets is typically small.
Each type corresponds to a single node in the intersection graph, so that a constant number of types corresponds to a constant number of nodes. That is, both of the lemma conditions are satisfied, even though only one is required.

Second, auditing is often localized. For instance, when considering audits performed by corporate managers, one would expect these managers 
to primarily inspect 
the activities of their direct subordinates. 
%Localized auditing does not   
%allow us to collapse targets inspected by different managers, because   
%of the different sets of inspection resources that each target can      
%be inspected by. Hence, the intersection graph can have many nodes.      
%However, 
This means that the inspection resources (inspection actions of a manager) auditing a node (activities of its subordinates) are disjoint from the 
inspection resources auditing any other node. Thus, our intersection    
graph has no edges, and the second lemma condition is satisfied. Slightly more complex situations, where, for instance, employees' activities are audited by two different managers, still satisfy the second condition.

%These scenarios can easily be further extended while still satisfying one of the lemma's two conditions. For example, we can have localized auditing, with central auditors that audit targets that the managers cannot audit (all such targets collapse into one node); or settings where each employee can be audited by two different managers.
%(1) central auditors that audit targets that the managers are not capable of auditing (all such targets collapse into one node) or (2) a sub-ordinate being managed by a 
%two managers (nodes with degree two) or (3) a constant number of sub-ordinates managed
%by at-most a constant number of managers (constant max degree and constant number of nodes with degree three or more). These variations capture many realistic scenarios.

\vspace{-1pt}

\subsection{Target-Specific Punishments} We present a brief overview of target-specific
punishments with the details in the appendix.
We extend our model to target-specific punishments by augmenting the program $P_n$: we use individual 
punishment levels $x_1, \ldots, x_n$, instead of using the same punishment $x$ for 
each target. The new optimization problem $PX_{n}$ differs from $P_n$ only in
(1) objective: $\max_{p_{i},x}  p_{n} \Delta_{D,n} - \sum_{j \in \{1, \ldots, n\}} a_j x_j$ and (2) quadratic constraints: $p_i (-x_i - \Delta_i) + p_n (x_n + \Delta_n) + \delta_{i,n}  \leq 0$.
%$$
%\begin{array}{llc}
%\displaystyle\max_{p_{i},x} & p_{n} \Delta_{D,n} - \sum_{j \in \{1, \ldots, n\}} a_j x_j \ ,&\\
%\mbox{s.t.}& \forall i \neq n. \ p_i (-x_i - \Delta_i) + p_n (x_n + \Delta_n) + \delta_{i,n}  \leq 0 & \\
%%&  \forall i \neq *. \ p_{i}(U^a_A (t_i) -x) + (1 - p_{i})U^u_A (t_i)  &\\
%%&  \ \ \ \ \ \ \ \ \ \ \leq p_{*}(U^a_A (t_n) - x) + (1 - p_{n})U^u_A (t_n)  \ ,&\\
%& \forall j. \ 0 \leq \sum_{i=1}^n p_i^j \leq 1 \ , \forall i. \ 0 \leq p_i = \sum_{j=1}^k p_i^j \leq 1 \ , &\\
%& \ \forall (j,i). \ p_i^j \geq  0 \ , \forall (j,i) \in R. \  p_{i}^j = 0 \ , \\
%& 0 \leq x \leq 1 \ . & 
%\end{array}
%$$
%Note that the defender's penalty term is now a linear combination of the target-specific punishment level.

The na\"ive way of discretizing each of the variables $x_1, \ldots, x_n$
and solving the resulting sub-problems is not polynomial time. 
Nevertheless, we show that 
it is possible to design an FPT approximation algorithm by discretizing only 
$p_n$, and casting the resulting sub-problems as  
second-order cone problems (SOCP), which can be solved in polynomial time~\cite{boyd2004convex}.
We first present the following intuitive result:
\begin{lemma} \label{targetlemma}
At the optimal point for $PX_n$, $x_n$ is always 0. Further, discretizing values of $p_n$ with interval size $\epsilon$ and solving resulting sub-problems yields a $\Theta(\epsilon)$ approximation.
\end{lemma}
The proof is in the appendix.
Next, we show that for fixed values of $x_n$ and $p_n$,  $PX_{n}$  reduces to an SOCP. We first describe a general SOCP problem (with variable $y \in \mathbb{R}^n$) that maximizes a linear objective $f^Ty$ subject to linear constraints $Fy = g$ and $m$ quadratic constraints of the form
$$
 \forall i \in \{1, \ldots, m\}.~||A_iy + b_i||_2 \leq c_i^T y + d_i 
$$
In particular, the constraint  
$
 k^2/4 \leq y_i(y_j + k')
 $
 is the same as $||[k~~(y_i - y_j - k')]^T||_2 \leq y_i + y_j + k$, which is an instance
 of the quadratic constraint above for appropriate $A,b, c, d$.

 % A special case of the above general problem, which we use, is
%when $A, b$ is chosen such that $Ay + b = [k~~(y_i - y_j - k')]^T$ for some constants
%$k, k'$, and $c, d$ chosen such that $c^y + d = y_i + y_j + k$. Note that it is always 
%possible to choose such $A,b,c,d$. Then, the second-order inequality 
%$
%||Ay + b||_2 \leq c^T y + d 
 %$
 %is same as
% $$
% k^2/4 \leq y_i(y_j + k')
% $$
Our problem can be cast as an SOCP by rewriting the quadratic constraints
as
$$
  p_n (x_n + \Delta_n) + \delta_{i,n}  \leq p_i (x_i + \Delta_i)
$$
Using our approach (discretizing $p_n$, $x_n = 0$) the LHS of the above inequality is a constant. If the
constant is negative we can simply throw out the constraint --- it is a tautology since the RHS is always positive. If
the constant is positive, we rewrite the constraint as a second-order constraint as 
described above (e.g., set $k = 2\sqrt{  p_n (x_n + \Delta_n) + \delta_{i,n} }$ and set $k\rq{} = \Delta_i$). The rest of the constraints are linear. Thus, the problem for each fixed value of
$p_n, x_n$ is an SOCP. Putting everything together, we get the following result.
\begin{theorem} \label{tsp}
The method described above is an FPT additive $\Theta(\epsilon)$-approximate algorithm for solving $PX_n$.
\end{theorem}

%The extension to multiple attacks is immediate; the only change required would be 
%to discrete $p_j, x_j$ for $j$ belonging to the index set of the $m$ targets under 
%consideration. Then, the resultant problem for each fixed values of $p_j, x_j$ for $j$ belonging to the index set of the $m$ targets can be cast as a SOCP similar to the single 
%attack case.

%Allowing different punishments for different targets reflects the common law 
%principle of punishment proportional to the crime~\cite{Proportional}. In our case, we allow the 
%defender to choose the best cost optimal punishment levels that deter the adversary.

\section{Fully Polynomial Time Approximation}

Our goal in this section is to develop an FPTAS for problem $P_n$, under the condition that the set $C$ returned by $\mathsf{CONSTRAINT\_FIND}$  has polynomially many constraints.
Our algorithm builds on an earlier algorithm~\cite{bcdps2013audit} for the restricted auditing 
scenario with just one defender resource. Since we solve the defender's problem after extracting the constraints $C$, the variables
in our problem are just the $p_i$'s and $x$. 
%The final solution involves computing feasible values of
%$p_i^j$'s (this can be done in polynomial time since it is a linear feasibility problem), followed
%by the Birkhoff-von Neumann decomposition (see details shown in Section~\ref{modelmultinspection}) to obtain the distribution
%over pure allocations.
%First, we simplify our problem $P_*$ by eliminating variables $p_i^j$ and introducing constraints in $C$ ($c \in C$), and (wlog) letting $*$ be $n$.
$$
\begin{array}{ll}
\displaystyle\max_{p_{i},x} & p_{n} \Delta_{D,n} - ax \ ,\\
\mbox{subject to}& \forall i \neq n. \\
&  p_i (-x - \Delta_i) + p_n (x + \Delta_n) + \delta_{i,n}  \leq 0 \ ,  \\
%&  \forall i \neq *. \ p_{i}(U^a_A (t_i) -x) + (1 - p_{i})U^u_A (t_i)  &\\
%&  \ \ \ \ \ \ \ \ \ \ \leq p_{*}(U^a_A (t_n) - x) + (1 - p_{n})U^u_A (t_n)  \ ,&\\
& c \in C, \ \ \forall i.  \ 0 \leq p_i \leq 1 \ , \ \ \ 0 \leq x \leq 1 \ .
\end{array}
$$

\medskip

\noindent\textit{Property of optimal points.} We state the following property of some
optimal points $p_i^*$'s and $x^*$ of the optimization:
\begin{lemma} \label{optimalprop}
There exists optimal points $p_i^*$'s and $x^*$ such that if
$p_n^* (x^* + \Delta_n) + \delta_{j,n} \geq 0$ then $p_n^* (x^* + \Delta_n) + \delta_{j,n} = p_j^* (x^* + \Delta_j)$ (i.e., quadratic constraint is tight) else when
$p_n^* (x^* + \Delta_n) + \delta_{j,n} < 0$ then $p_j^* = 0 $.
\end{lemma}
\begin{proof}[Proof Sketch]
The quadratic constraint can be written as $\frac{p_n (x + \Delta_n) + \delta_{j,n}}{x + \Delta_i} \leq p_i$
At the optimal point if $p_n^* (x^* + \Delta_n) + \delta_{j,n} \geq 0$ then if we have $\frac{p_n^* (x^* + \Delta_n) + \delta_{j,n}}{x + \Delta_i} < p_i^*$, we can always reduce $p_i^*$ without affecting the 
objective value till we get an equality. Also, for the case $p_n^* (x^* + \Delta_n) + \delta_{j,n} < 0$ we 
can reduce $p_i^*$ to $0$.
\end{proof}

%\medskip
We focus on finding one of the optimal points with the property stated above. Next, we sort $\delta_{i,n}$'s to get a sorted array 
$\delta_{(i),n}$ in ascending order. Then, we split the optimization
problem $P_n$ into sub-problems $EQ_{(j)}$, where in each problem $EQ_{(j)}$ it is assumed that 
$p_n, x$ lies between the hyperbolas $p_n (x + \Delta_n) + \delta_{(j),n}$ (open 
boundary) and $p_n (x + \Delta_n) + \delta_{(j+1),n}$ (closed boundary) in the plane 
spanned by $p_n, x$. Thus, in $EQ_{(j)}$, $p_n 
(x + \Delta_n) + \delta_{(j),n}$ is non-negative
for all $(k) > (j)$ and negative otherwise. Using the property
of the optimal point above, for the non-negative case we obtain equalities for the quadratic constraints and for the negative case we claim that for all $(k) \leq (j)$ we can set $p_{(k)} = 0$.
The optimal value for $P_n$ can be found
by solving each $EQ_{(j)}$ and taking the best solution from these sub-problems.

The optimization problem for $EQ_{(j)}$ is as follows:
 $$
\begin{array}{ll}
\displaystyle\max_{p_{i},x} & p_{n}\Delta_{D,n} - ax \ ,\\
\mbox{subject to}& \forall (i) > (j). \ 0 \leq \frac{p_n (x + \Delta_n) + \delta_{(i),n}}{x + \Delta_{(i)}} = p_{(i)} \leq 1  \\
& p_n (x + \Delta_n) + \delta_{(j),n} < 0 \\
& p_n (x + \Delta_n) + \delta_{(j+1),n} \geq 0 \\
%&  \forall i \neq *. \ p_{i}(U^a_A (t_i) -x) + (1 - p_{i})U^u_A (t_i)  &\\
%&  \ \ \ \ \ \ \ \ \ \ \leq p_{*}(U^a_A (t_n) - x) + (1 - p_{n})U^u_A (t_n)  \ ,&\\
& c \in C ,\ \forall (i) \leq (j).  \ p_{(i)} = 0 \ , \ \ \ 0 \leq x \leq 1 \ .
\end{array}
$$
As no $p_i$ (except $p_n$) appears in the objective, and due to the equality constraints
on particular $p_i$'s, we can replace those $p_i$'s by a function of $p_n, x$. 
Other $p_i$'s are zero. Next, by a series of simple algebraic manipulations we obtain 
a two-variable optimization problem:
 $$
\begin{array}{llc}
\displaystyle\max_{p_{i},x} & p_{n}\Delta_{D,n} - ax \ ,&\\
\mbox{subject to}& \forall b \in \{1, \dots,  B\}. ~ p_n \leq f_b(x) & \\
%& p_n (x + \Delta_n) + \delta_{(j),n} < 0 &\\
%& p_n (x + \Delta_n) + \delta_{(j)+1,n} \geq 0 &\\
%&  \forall i \neq *. \ p_{i}(U^a_A (t_i) -x) + (1 - p_{i})U^u_A (t_i)  &\\
%&  \ \ \ \ \ \ \ \ \ \ \leq p_{*}(U^a_A (t_n) - x) + (1 - p_{n})U^u_A (t_n)  \ ,&\\
& \ 0 \leq p_n \leq 1 \ , \ \ \ 0 \leq x \leq 1 \ ,
\end{array}
$$
where $B$ is the total number of constraints, which is
of the same order as $|C|$. The details of the algebraic steps are in the full version.

\medskip
\noindent\textit{Solving the sub-problem.} Our two final lemmas are not hard to prove. Their proofs appear in the full version.
 \begin{lemma} \label{fixedxpn}
 Problem $EQ_{(j)}$ can be solved efficiently for a fixed value $x$ or fixed value of $p_n$.
  \end{lemma} 
   \begin{lemma} \label{boundary}
 The optimal point for $EQ_{j}$ cannot be an interior
point of the region defined by the constraints, i.e., at least one of the inequalities is tight for the optimal point.
  \end{lemma} 
  \begin{proof}[Proof Sketch]
  It is easy to see that if all constraints are non-tight at the optimal point, then $p_n$ can be increased by a small amount without violating any constraint and also increasing the objective value. Thus, some constraint
  must be tight.
  \end{proof}
\begin{algorithm}[h] \label{algmult1}
\DontPrintSemicolon
Solve the problem for $p_n= 0,1$ and $x = 0,1$\; 
Collect solutions $(p_n, x)$ from the above in $M$\;
\For{$b \leftarrow 1$ \KwTo $B$}{
	Replace $p_n = f_b(x)$ in the objective to get $F(x) = f_b(x)\Delta_{D,n} - ax$ \;
	Take the derivative to get $F'(x) = \frac{\partial F(x)}{\partial x}$\;
	$R \leftarrow \mathsf{ROOTS}(F'(x), 2^{-l}, (0,1))$\;
	$R' \leftarrow \mathsf{MAKEFEASIBLE}(R)$\;
	From $R'$ obtain set $M'$ of potential solutions $(p_n, x)$\;
	$M \leftarrow M \cup M'$\;
}
\For{$(b,b') \in \left\{(b,b') ~|~ b \in B, b \in B, b' > b \right\}$}{
	Equate $f_b(x) = f_{b'}(x)$ to get $F(x) = 0$ \;
	$R \leftarrow \mathsf{ROOTS}(F(x), 2^{-l}, (0,1))$\;
	$R' \leftarrow \mathsf{MAKEFEASIBLE}(R)$\;
	From $R$ obtain set $M'$ of potential solutions $(p_n, x)$\;
	$M \leftarrow M \cup M'$\;
}
$(p_n^*, x^*) \leftarrow \arg\max_{M} \{ p_n \Delta_{D,n} - ax \}$ \;
\Return $(p_n^*, x^*)$
\caption{$\mathsf{APX\_SOLVE}(l, EQ_{(j)})$}
\end{algorithm}

We are now ready to present the FPTAS, given as  
Algorithm~\ref{algmult1}. The algorithm first searches potential optimal points on the boundaries (in the first loop)
and then searches potential optimal points at the intersection of two boundaries (second loop). %\ap{I don't understand what you mean by ``searches the maximizer''.}
The roots are found to an additive approximation factor of $2^{-l}$ in time 
polynomial in the size of the problem representation and $l$~\cite{schonhage1982fundamental}. As shown in Blocki et al.~\shortcite{bcdps2013audit}, the case of roots lying outside the feasible region (due
to approximation) is taken care of by the function $\mathsf{MAKEFEASIBLE}$. The first 
loop iterates a maximum of $n$ times, and the second loop iterates a maximum of
$n \choose 2$ times.
Thus,
we have the following result. 
\begin{theorem}
The optimization problem $P_n$ can be solved with an additive approximation
factor of $\Theta(2^{-l})$ in time polynomial in the input size and $l$, i.e., our algorithm
to solve $P_n$ is an FPTAS.
\end{theorem}

%\subsection{Optimization algorithm}
%\input{Optimization}
%
%\subsection{Finding constraint set $C$}
%\input{Constraint}

\section{Experimental Results}

%\ad{How did you pick the numbers in Table 1? What are comparable numbers from related security game experiments?}
In this section, we empirically demonstrate the speedup gains from our optimization transformation
for both audit games and security games. We obtained speedups of up to $3\times$
for audit game instances and over $100\times$ for associated security game instances.

Our experiments were run on a desktop with quad core 3.2 GHz processor and 6GB RAM.
Code was written in Matlab using the built-in large scale
interior point method implementation of linear programming. 
%Our 
%implementations are not optimized for speed, but only serve to compare the
%naive approach with our reduced problem size approach.
We implemented two FPT algorithms---with and without the optimization 
transformation. For both the algorithms, we used
the same problem inputs in which utilities were generated 
randomly from the range $[0, 1]$, $a$ was fixed to $0.01$, $x$ was 
discretized with interval size of $0.005$.
% and it was allowed to vary from $0$ to $10$
%(the upper bound of $10$, instead of $1$, is not important as utilities can 
%be scaled accordingly). 

%\begin{figure}[t]
%\begin{center}
%\includegraphics[scale=0.50]{100_plot.pdf}
%\end{center}
%\caption{FPT algorithm running time with 100 target, 10 inspection resources. Green marks for Figure~\ref{200plot} problem case, blue with the grid constraints} \label{100plot}
%\end{figure}

\renewcommand{\arraystretch}{1.0}
\begin{table}[t]
\begin{tabular}{l r r r r}
\toprule
& &  \#Resource & \multicolumn{2}{c}{Time (min)} \\
\cmidrule(r){4- 5}
Game & \#Target  & (GroupSize)  & T & NT~ \\
\midrule
Audit &   100   & 10 (2)    & 12   & 15~  \\
 Audit &     200    & 100 (10)       & 81    &  234~ \\
Security &   3,000     & 500 (10)    & 28   &   8,500\footnotemark[1] \\
Security &    5,000   & 1,000  (20)   & 119  &  110,000\footnotemark[1] \\
\bottomrule
\end{tabular}
\caption{FPT algorithm running times (in min) with our optimization transformation (T) and no transformation (NT).} 
 \label{experimentresult}
\end{table}
\footnotetext[1]{These data points are extrapolations from runs that were allowed to run 12 hours. The extrapolation was based
on the number of optimization problems solved vs the total number of optimization
problems (3000/5000 total problems for 3000/5000 targets).} 

%\begin{figure}[t]
%\begin{center}
%\includegraphics[scale=0.50]{200_plot.pdf}
%\end{center}
%\caption{FPT algorithm running time with 200 target, 100 inspection resources. Green marks for transformed problem case, blue with the grid constraints.} \label{200plot}
%\end{figure}

%We ran two experiments: (1) with 100 targets and 10 inspection resources. The 10 resources
%were divided into groups of two resources each, with each group of resource capable of
%inspecting 20 of the targets, (2) with 200 targets and 100 inspection resources. The 100 resources
%were divided into groups of 10 resources each, with each group of resource capable of
%inspecting 20 of the targets.
%
%The targets that any group of resources were capable
%of inspecting was disjoint from the target set for any other group.
%Our experiments ran for 
%10 random instances and 5 random instances in the two cases respectively. The results are 
%shown in Figure~\ref{100plot} and Figure~\ref{200plot}, with green marks for the transformed
%scenario. 

We ran experiments for audit games and security games with 
varying number of targets and resources. The resources were divided into equal sized groups such that the targets any group of resources could inspect was disjoint from the target set for any other group. 
Table~\ref{experimentresult} %\ad{Table number not showing up when I compile. Check.} 
shows our results with the varying number of targets, resources and size of the group of
targets. The results are an average over 5 runs of the optimization with random utilities in each run.
%Our second experiment used  The targets that any group of resources were capable
%of inspecting was disjoint from the target set for any other group. We sampled
%5 random instances of the problem. The results are 
%shown in Figure~\ref{200plot}, with green marks for the reduced problem
%scenario. In addition to the reduction in problem size leading to
%better time performance, we observe that the improvement is much greater for this
%bigger problem than in the case with 100 targets and 10 resources.
Audit games take more time to solve than corresponding security games with 
a similar number of targets as we run the
corresponding LP optimization 200 times (the discrete interval for $x$ is 0.005). Hence we solve for larger security game instances. 

Our implementations did not optimize for speed using heuristics because our goal was to only test the speedup gain from our optimization transformation. Thus, we do not scale up
to the number of targets that heuristic approaches such as {\sc ORIGAMI}~\cite{Kiekintveld1558108} achieve (1,000,000 targets). {\sc ORIGAMI}
works by iteratively building the optimal attack set, i.e., the set of targets that the adversary finds most attractive at the optimal solution point.
However, {\sc ORIGAMI} considers only coverage probabilities (marginals). Thus, its 
output may not be implementable with scheduling constraints on resources. In contrast, our
approach guarantees that the coverage probabilities output are implementable. Wedding the two 
approaches to obtain scalable and provably implementable audit and security game solutions 
remains an interesting direction for future work.  
\newpage
\bibliography{References} 
\bibliographystyle{aaai} 

%\newpage
\appendix
\section*{Appendix}

\section{Birkhoff-von Neumann}
 For the sake of completeness, we state the 
Birkhoff-von Neumann theorem from Korzhyk et al.~\cite{KorzhykCP10}, that is used to 
decompose the probability matrix into pure actions efficiently.

\textit{
(Birkhoff-von Neumann~\cite{MR0020547}). Consider
an $m \times n$ matrix $M$ with real numbers $a_{ij} \in [0, 1]$,
such that for each $1 \leq i \leq m$, $\sum^n_{j=1} a_{ij} \leq 1$, and for
each $1 \leq j \leq n$, $\sum^m_{i=1} a_{ij} \leq 1$. Then, there exist matrices $M_1,
M_2, \ldots , M_q$, and weights $w_1,w_2, \ldots ,w_q \in (0, 1]$,
such that:
\begin{enumerate}
\item $\sum^q_{k=1} w_k = 1$
\item $\sum^q_{k=1} w_k M_k = M$
\item for each $1 \leq k \leq q$, the elements of $M_k$ are $a^k_{ij} \in \{0, 1\}$
\item for each $1 \leq k \leq q$, we have: for each $1 \leq i \leq m$,
$\sum^n_{j=1} a^k_{ij} \leq 1$, and for each $1 \leq j \leq n$, $\sum^m_{i=1} a^k_{ij} \leq 1$.
\end{enumerate}
Moreover, $q$ is $O((m + n)^2)$, and the $M_k$ and $w_k$ can be
found in $O((m + n)^{4.5})$ time using Dulmage-Halperin algorithm~\cite{MR0078627}.
}

Clearly, our variables $p_{i}^j$'s can be considered as the matrix $M$ in the result
above, and hence can be decomposed into pure actions efficiently.

\section{Obtaining the two variable optimization}
The optimization problem for $EQ_{(j)}$ is as follows:
 $$
\begin{array}{llc}
\displaystyle\max_{p_{i},x} & p_{n}\Delta_{D,n} - ax \ ,&\\
\mbox{subject to}& \forall (i) > (j). \ 0 \leq \frac{p_n (x + \Delta_n) + \delta_{(i),n}}{x + \Delta_{(i)}} = p_{(i)} \leq 1 & \\
& p_n (x + \Delta_n) + \delta_{(j),n} < 0 &\\
& p_n (x + \Delta_n) + \delta_{(j+1),n} \geq 0 &\\
%&  \forall i \neq *. \ p_{i}(U^a_A (t_i) -x) + (1 - p_{i})U^u_A (t_i)  &\\
%&  \ \ \ \ \ \ \ \ \ \ \leq p_{*}(U^a_A (t_n) - x) + (1 - p_{n})U^u_A (t_n)  \ ,&\\
& c \in C ,\ \forall (i) \leq (j).  \ p_{(i)} = 0 \ , \ \ \ 0 \leq x \leq 1 \ .
\end{array}
$$
As no $p_i$ (except $p_n$) shows up in the objective, and due to the equality constraints
on particular $p_i$'s, we can replace those $p_i$'s by a function of $p_n, x$. 
Other $p_i$'s are zero.
Denote by $c(p_n, x)$ the inequality
obtained after substituting $p_i$ with the function of $p_n, x$ (or zero) in the constraint
$c \in C$. Thus, we get the 
following two variable optimization problem
 $$
\begin{array}{llc}
\displaystyle\max_{p_{i},x} & p_{n}\Delta_{D,n} - ax \ ,&\\
\mbox{subject to}& \forall (i) > (j). \ \frac{p_n (x + \Delta_n) + \delta_{(i),n}}{x + \Delta_{(i)}} \leq 1 & \\
& p_n (x + \Delta_n) + \delta_{(j),n} < 0 &\\
& p_n (x + \Delta_n) + \delta_{(j+1),n} \geq 0 &\\
%&  \forall i \neq *. \ p_{i}(U^a_A (t_i) -x) + (1 - p_{i})U^u_A (t_i)  &\\
%&  \ \ \ \ \ \ \ \ \ \ \leq p_{*}(U^a_A (t_n) - x) + (1 - p_{n})U^u_A (t_n)  \ ,&\\
&  c(p_n, x) \in C , \ 0 \leq p_n \leq 1 \ , \ \ \ 0 \leq x \leq 1 \ .
\end{array}
$$
Observe that we removed the $0 \leq$ condition in the first set of constraints because that is 
implied by the next two constraints. Next, note that any constraint (indexed by $b$) with two 
variables $p_n, x$ can be expressed as $p_n \leq f_b(x)$ (closed boundary) 
for constraint-specific ratio of polynomials $f_b(x)$ or by 
$p_n < \frac{-\delta_{(j),n}}{x + \Delta_n}$ (open boundary). However, we close the
open boundary, i.e., $p_n \leq \frac{-\delta_{(j),n}}{x + \Delta_n}$, and solve the 
problem, returning an infeasible solution in case the solution is on the boundary
$p_n = \frac{-\delta_{(j),n}}{x + \Delta_n}$. This is justified by the fact that the curve
$p_n = \frac{-\delta_{(j),n}}{x + \Delta_n}$ is included in the other optimization 
problem $EQ_{(j-1)}$, and would be output by that sub-problem if it indeed is the 
global maximizer. Thus, we represent the optimization problem as
 $$
\begin{array}{llc}
\displaystyle\max_{p_{i},x} & p_{n}\Delta_{D,n} - ax \ ,&\\
\mbox{subject to}& \forall b \in \{1, \dots,  B\}. ~ p_n \leq f_b(x) & \\
%& p_n (x + \Delta_n) + \delta_{(j),n} < 0 &\\
%& p_n (x + \Delta_n) + \delta_{(j)+1,n} \geq 0 &\\
%&  \forall i \neq *. \ p_{i}(U^a_A (t_i) -x) + (1 - p_{i})U^u_A (t_i)  &\\
%&  \ \ \ \ \ \ \ \ \ \ \leq p_{*}(U^a_A (t_n) - x) + (1 - p_{n})U^u_A (t_n)  \ ,&\\
& \ 0 \leq p_n \leq 1 \ , \ \ \ 0 \leq x \leq 1 \ ,
\end{array}
$$
where $B$ is the total number of constraints described as above. 

\section{Target-Specific Punishments} \label{sec:TargetSpecificPunishment}
In this section we extend our FPT result to target-specific punishments. While
the formulation of the extended problem is not that different, solving it becomes significantly more challenging. 

In more detail, we augment the program $P_n$ by using individual 
punishment levels $x_1, \ldots, x_n$, instead of using the same punishment $x$ for 
each target. The new optimization problem $PX_{n}$ is
$$
\begin{array}{llc}
\displaystyle\max_{p_{i},x} & p_{n} \Delta_{D,n} - \sum_{j \in \{1, \ldots, n\}} a_j x_j \ ,&\\
\mbox{s.t.}& \forall i \neq n. \ p_i (-x_i - \Delta_i) + p_n (x_n + \Delta_n) + \delta_{i,n}  \leq 0 & \\
%&  \forall i \neq *. \ p_{i}(U^a_A (t_i) -x) + (1 - p_{i})U^u_A (t_i)  &\\
%&  \ \ \ \ \ \ \ \ \ \ \leq p_{*}(U^a_A (t_n) - x) + (1 - p_{n})U^u_A (t_n)  \ ,&\\
& \forall j. \ 0 \leq \sum_{i=1}^n p_i^j \leq 1 \ , \forall i. \ 0 \leq p_i = \sum_{j=1}^k p_i^j \leq 1 \ , &\\
& \ \forall (j,i). \ p_i^j \geq  0 \ , \forall (j,i) \in R. \  p_{i}^j = 0 \ , \\
& 0 \leq x \leq 1 \ . & 
\end{array}
$$
Note that the defender's penalty term is now a linear combination of the target-specific punishment level.

The na\"ive way to approach the above problem is to discretize each of the variables $x_1, \ldots, x_n$
and solve the resulting sub-problems, which are linear programs. However, such a discretization of size $\epsilon$,
even if yielding a $\Theta(\epsilon)$ additive approximation, will run in time $O((1/\epsilon)^n)$, which is not polynomial for constant $\epsilon$. 

Nevertheless, we show that 
it is possible to design an algorithm that runs in time polynomial in $\epsilon$ and yields
a $\Theta(\epsilon)$ additive approximation, by discretizing only 
$p_n$, and casting the resulting sub-problems as second-order cone programs 
(SOCP), which can be solved in polynomial time~\cite{boyd2004convex}. Letting $S(n)$ denote the (polynomial) running time of the SOCP, we 
obtain a running time of $O(S(n)/\epsilon)$, given fixed bit precision.

We first present the following intuitive result:
\medskip

\textit{Restatement of Lemma~\ref{targetlemma}.
At the optimal point for $PX_n$, $x_n$ is always 0. Further, considering discrete values of $p_n$ with interval size $\epsilon$ and solving the resulting sub-problems exactly yields a $\Theta(\epsilon)$ approximation.
}

\medskip
The proof is in the Missing Proofs section of the appendix.

Next, we show that for fixed values of $x_n$ and $p_n$ the optimization problem  $PX_{n}$  reduces to a second-order cone program. We first describe a general SOCP problem:
$$
\begin{array}{llc}
\displaystyle\max_{y} & f^T y \ ,&\\
\mbox{subject to}& \forall i \in \{1, \ldots, m\}.~||A_iy + b_i||_2 \leq c_i^T y + d_i & \\
%&  \forall i \neq *. \ p_{i}(U^a_A (t_i) -x) + (1 - p_{i})U^u_A (t_i)  &\\
%&  \ \ \ \ \ \ \ \ \ \ \leq p_{*}(U^a_A (t_n) - x) + (1 - p_{n})U^u_A (t_n)  \ ,&\\
& Fy = g \ .
\end{array}
$$
where the optimization variable is $y \in \mathbb{R}^n$, and all constants are of 
appropriate dimensions. In particular, we observe that the constraint  
$$
 k^2/4 \leq y_i(y_j + k')
 $$
 can always be written as a second order inequality $||Ay + b||_2 \leq c^T y + d  $ by selecting $A$ and $b$ such that $Ay + b = [k~~(y_i - y_j - k')]^T$
 and $c,d$ such that $c^T y + d = y_i + y_j + k$.

 % A special case of the above general problem, which we use, is
%when $A, b$ is chosen such that $Ay + b = [k~~(y_i - y_j - k')]^T$ for some constants
%$k, k'$, and $c, d$ chosen such that $c^y + d = y_i + y_j + k$. Note that it is always 
%possible to choose such $A,b,c,d$. Then, the second-order inequality 
%$
%||Ay + b||_2 \leq c^T y + d 
 %$
 %is same as
% $$
% k^2/4 \leq y_i(y_j + k')
% $$
Our problem can be cast as a SOCP by rewriting the quadratic constraints
as
$$
  p_n (x_n + \Delta_n) + \delta_{i,n}  \leq p_i (x_i + \Delta_i)
$$
Using our approach (discretizing $p_n$, $x_n = 0$) the LHS of the above inequality is a constant. If the
constant is negative we can simply throw out the constraint --- it is a tautology since the RHS is always positive. If
the constant is positive, we rewrite the constraint as a second-order constraint as 
described above (e.g., set $k = 2\sqrt{  p_n (x_n + \Delta_n) + \delta_{i,n} }$ and set $k\rq{} = \Delta_i$). The rest of the constraints are linear, and can be rewritten in
equality form by introducing slack variables. Thus, the problem for each fixed value of
$p_n, x_n$ is an SOCP, and can be solved efficiently. Putting everything together, we get the main Theorem~\ref{tsp} of the target-specific punishment section in the paper.

\section{Missing Proofs} \label{AppC}

\begin{proof}[Proof of Theorem~\ref{FPTLemma}]
%Discretize both $p_n$ and $x$ in intervals of $\epsilon$ and solve the resulting LP sub-
%problems. There are $\left ( \frac{1}{\epsilon} \right )^2$ such sub-problems. Thus, the 
%running time is $O(\mathsf{LP}(n) \left ( \frac{1}{\epsilon} \right )^2)$. It is enough to 
%show that the objective is approximated to an additive factor of $\theta(\epsilon)$. Note
%that we assume constant bi-precision, so each input number is considered a constant.
%Suppose $x^o, p_n^o$ is an optimal point (with proper values of other variables too).
%

The quadratic constraint can be rewritten as 
$$
x(p_n - p_i) - p_i \Delta_i + p_n \Delta_n + \delta_{i,n} \leq 0
$$

Suppose we have an optimal point $x^o$, $p_i^o$'s. First, we note that if $x^o = 1$ (resp. $x^o=0$) then our algorithm will find the optimal solution exactly because $x = 1$ (resp. $x=0$) is one of the discrete values of $x$ considered by our algorithm. In the rest of the proof we assume that $0 < x^o < 1$. Let $\epsilon\rq{} > 0$ denote the smallest value s.t. $x^o+\epsilon\rq{}$ is one of the discrete values of $x$ considered by our algorithm (e.g. $x^0+\epsilon\rq{}$ $\in \{\epsilon i ~\vline~i \in \mathbb{N}\} \cup \{1\}$). We let 
\[ \tau = \max\left\{0, \max_{i \leq n} \frac{\epsilon\rq{}\left(p_n^o-p_i^o \right)}{x^o + \epsilon\rq{}+\Delta_n} \right\} \ ,\]
and we consider two cases. \\
{\bf Case 1:} $p_n^o \geq \tau $. In this case we set $\hat{x} = x^o+\epsilon\rq{}$, $\hat{p}_n = p_n^o-\tau$ and $\hat{p}_i = p_i^o$ for each $i \neq n$. We first show that this is a valid solution. For each $i \neq n$ we have
\begin{eqnarray*}
& & \hat{x}(\hat{p}_n -\hat{ p}_i) - \hat{p}_i \Delta_i + \hat{p}_n \Delta_n + \delta_{i,n} \\
 &=& (x^o+\epsilon\rq{})(p_n^o -\tau- p_i^o) - p_i^o \Delta_i + (p_n^o-\tau) \Delta_n + \delta_{i,n} \\
&=& \left( x^o(p_n^o - p_i^o) - p_i^o \Delta_i + p_n^o \Delta_n + \delta_{i,n} \right)  \\ 
& &+ \left(\epsilon\rq{}(p_n^o-p_i^o)-\tau\epsilon\rq{} - \tau \Delta_n \right)  - x^o \tau \\
&\leq& \epsilon\rq{}(p_n^o-p_i^o) - \tau (x^o + \Delta_n+\epsilon\rq{}) \\
&=&  \epsilon\rq{}(p_n^o-p_i^o) \\
&\ \ \ \ \ -&  (x^o + \Delta_n+\epsilon\rq{})\max\left\{0, \max_{k \leq n} \frac{\epsilon\rq{}\left(p_n^o-p_k^o \right)}{ x^o + \Delta_n+\epsilon\rq{}} \right\}  \\
&\leq& \epsilon\rq{}(p_n^o-p_i^o) - \max_{k \leq n} \epsilon\rq{}\left(p_n^o-p_k^o \right) \\
&\leq& 0 \ .
\end{eqnarray*}
Note that the loss in utility for the defender is  $\Delta_{D,n}\tau + a \epsilon'$ which is upper bounded by 
\[ \epsilon \left( \max_{i \leq n} \frac{\left|p_n^o-p_i^o \right|}{x^o + \Delta_n} +a  \right) \ . \]
That is  $\epsilon b$ for a constant $b$ when either $x^o$ is greater than a constant or $\Delta_n \neq 0$. Observe that due to the fixed bit precision assumption $\Delta_n \neq 0$ implies $\Delta_n$ is greater than a constant.  We remark that the solution returned by our algorithm will be at least as good as the solution given by $\hat{x}$ and $\hat{p}_i$\rq{}s because $\hat{x}$ is one of the discrete values of $x$ considered. Thus, the defender\rq{}s utility in the solution given by our algorithm is $\theta(\epsilon)$-close to the defenders utility in the optimal solution. \\

{\bf Case 2:} $p_n^o < \tau$. In this case we set $\hat{p_n} = 0$, $\hat{x} = x^o + \epsilon\rq{}$ and $\hat{p_i} = p_i^o$ for each $i \neq 0$. We first show that this is a valid solution. For each $i \neq n$ we have 
\begin{eqnarray*}
& & \hat{x}(\hat{p}_n -\hat{ p}_i) - \hat{p}_i \Delta_i + \hat{p}_n \Delta_n + \delta_{i,n} \\
 &=& (x^o+\epsilon\rq{})(p_n^o-p_n^o- p_i^o) - p_i^o \Delta_i + (p_n^o-p_n^o) \Delta_n + \delta_{i,n} \\
&=& \left( x^o(p_n^o - p_i^o) - p_i^o \Delta_i + p_n^o \Delta_n + \delta_{i,n} \right)  \\ 
& &+ \epsilon\rq{}(-p_i^o) - x p_n^o-p_n^o\Delta_n \\
&\leq& 0 \ .
\end{eqnarray*}
Note that the loss in utility for the defender is at most  $\Delta_{D,n}\tau + a \epsilon'$ which is upper bounded by 
\[ \epsilon \left( \max_{i \leq n} \frac{\left|p_n^o-p_i^o \right|}{x^o + \Delta_n} +a  \right) \ . \]
That is  $\epsilon b$ for a constant $b$.   We  again remark that the solution returned by our algorithm will be at least as good as the solution given by $\hat{x}$ and $\hat{p}_i$\rq{}s because $\hat{x}$ is one of the discrete values of $x$ considered. Thus, the defender\rq{}s utility in the solution given by our algorithm is $\theta(\epsilon)$-close to the defenders utility in the optimal solution. \\
\end{proof}

\begin{proof}[Proof of Lemma~\ref{reduction}] 
Since the optimization objective depends on the variables $p_i$'s only, and the quadratic constraints are same in both problems we just need to show that 
the regions spanned by the variables $p_i$'s, as specified by the linear constraints, are the same in both problems.

First, let $p_i$'s, $p_i^j$'s belong to region given by the grid constraints. Then, by the 
definition of $C$, for any $c_{ML}$ we know $M$ is audited only by $L$. Thus, for $i \in
\{t_1, \ldots, t_{|M|}\}$ the variables $p_i^j$ are non-zero only when $j \in \{s_1, \ldots, s_{|L|}\}$. Therefore, 
$$\sum_{i \in
\{t_1, \ldots, t_{|M|}\}} p_i  \leq \sum_{i \in
\{t_1, \ldots, t_{|M|}\}} \sum_{j \in \{s_1, \ldots, s_{|L|}\}} p_i^j \leq |L|$$
Hence, $p_i$'s satisfies all constraints in $C \cup \{0 \leq p_i \leq 1\}$, and therefore
$p_i$'s belong to region given by $C \cup \{0 \leq p_i \leq 1\}$.

Next, let $p_i$'s belong to region given by $C \cup \{0 \leq p_i \leq 1\}$. We first show that the
extreme points of the convex polytope given by $C \cup \{0 \leq p_i \leq 1\}$ sets the 
variables $p_1, \ldots, p_n$ to either $0$ or $1$. 

\textbf{Extreme points 0/1}

We prove this result with additional (redundant) constraints in $C$. These redundant constraints are the ones that were dropped due to $|L| \geq |M|$. As these constraints are redundant, the polytope with or without them is same, and so are the extreme points. Thus, for this part (extreme points) we assume $C$ includes these redundant constraints. 
%Observe that the constraints can be written as
%$A p \leq b$, where $A$ is a 0/1 matrix with $n$ columns, $p$ is the vector $p_1, \ldots, p_n$ and $b$ is the integral vector
%with the RHS values of the constraints. Further note that the matrix $A$ has the form $A = \begin{bmatrix} A_C\\ I \\-I \end{bmatrix}$, where $I$ is the identity matrix. The identity matrices captures the inequalities
%$p_i \leq 1$ and $-p_i \leq 0$. 
%We will prove the total unimodularity of matrix $A$, which implies the extreme points are integral and given the $[0,1]$ bounds on $p_i$ we infer that any extreme point sets the 
%variables $p_1, \ldots, p_n$ to either $0$ or $1$. A totally unimodular matrix is defined as one in which
%any square sub-matrix has determinant either $0$, $-1$ or $+1$.
%First, we list some operations that preserve total unimodularity (well-known result).
%\begin{itemize}
%\item Removing any row or column.
%\item Adding a row or column containing only $0$'s, except one $1$.
%\item Multiplying a row or column by $-1$.
%\item Adding one more row or column already present in the matrix.
%\item Adding a row or column containing only $0$'s.
%\end{itemize}
%Then, it is sufficient to prove the total unimodularity of $A_C$ as the first three operations in the above list of operations implies the
%total unimodularity of $A$. 

We will do this by induction on the size of the
restricted set $R$. The base case is when $|R| = 0$, then the only constraint in $C$ is 
$\sum_i p_i \leq k$, i.e., $A_C$ is $\vec{1}$. It can be checked directly the extreme points have $k$ ones and other zeros.

Next assume the result holds for all $R$ with $|R| = w$. Consider a restriction $R$ with $|R| = w+1$. Choose any particular restriction, say $(j,i) \in R$. Suppose when this restriction does not exist, we have
the restriction $R' = R \backslash (j,i)$ of size $w$ and by induction hypothesis with $R'$ the extreme points are 0/1. We proceed to do the induction by checking how the constraint set $C$ and $C'$ differ.. We divide the proof into two cases. 

 \textbf{Case 1} First, suppose with restriction $R'$,
$t_i$ could be inspected only by $k_j$, then with $R$, $t_i$ is not inspected at all. Recall that any constraint $c_{M,L}$ denotes a set of resources $L$ and targets $M$ such that the targets are inspected only by resources. We construct the set of constraints $C$ and $C'$ by iterating through all subsets of resources. If for any subset of resources $L$, $k_j \notin L$ then the set of targets inspected only by $L$ does not include $t_i$ for both $R$ and $R'$ and is the same set for both $R$ and $R'$. Thus, 
$c_{M,L} \in C$ and $c_{M,L} \in C'$ and $p_i$ does not show up in these constraints. On the other hand, if $k_j \in L$, then for $R$ we have $t_i \notin M$ but for $R'$ we have $M' = M \cup t_i$. Thus, for such a subset $L$ we know that
\begin{itemize}
\item (\textsf{differ1}) $C$ and $C'$ include these constraint with the difference that $c_{M',L}$ has the additional variable $p_i$ on the LHS.
\end{itemize}
Thus, the constraints $C$ are formed from constraints $C'$ by setting $p_i = 0$. But, setting 
$p_i=0$ is the intersection of $p_i=0$ and the polytope given by $C'$. $p_i = 0$ is a $k-1$-face of the polytope given by $C'$. Thus, the extreme points of this intersection
polytope must be a subset of extreme points of polytope given by $C'$ and hence integral.

 \textbf{Case 2} In the second case with restriction $R'$,
$t_i$ is inspected by at least one resource other than $k_j$, then with restriction $R$, $t_i$ is still inspectable.
Let the set of resources that can inspect $t_i$ with restriction $R$ be $K_i$. We construct the set of constraints $C$ and $C'$ by iterating through all subsets of resources. If for any subset of resources $L$, $k_j \notin L$ then there are two cases possible:
\begin{itemize}
\item $K_i \subseteq L$ which implies $t_i \in M$ and $t_i \notin M'$ and $M = M' \cup t_i$. For this scenario we know that
\begin{itemize}
\item (\textsf{differ2}) $C'$ and $C$ both include these constraint with the difference that $c_{M,L}$ has the additional variable $p_i$ on the LHS.
\end{itemize}
\item $K_i \nsubseteq L$ which implies $t_i \notin M$ and $t_i \notin M'$, i.e., the set of targets inspected only by $L$ is same for both $R$ and $R'$
\end{itemize}
On the other hand, if $k_j \in L$, then there are two cases possible: 
\begin{itemize}
\item $K_i \subseteq L$ which implies $t_i \in M$ and $t_i \in M'$, i.e., the set of targets inspected only by $L$ is same for both $R$ and $R'$.
\item $K_i \nsubseteq L$ which implies $t_i \notin M$ and $t_i \notin M'$, i.e., the set of targets inspected only by $L$ is same for both $R$ and $R'$
\end{itemize}

For the cases where the set of targets inspected is same for both $R$ and $R'$, as argued earlier,
$c_{M,L} \in C$ and $c_{M,L} \in C'$.  Thus, the only scenario to reason about is \textsf{differ2}.
We do a step-by-step proof, by modifying constraints in $C'$ one by one to obtain the constraints $C$, in the
process showing for each step that the extreme points are 0/1. Thus, at each step we modify the polytope with 0/1 extreme points (say $R_{l}$) given by $C_l$ to obtain the polytope ($R_{l+1}$) given by $C_{l+1}$ by modifying one constraint $c$ that did
not have $p_i$ on the LHS to one $c_{+p_i}$ that has $p_i$ on the LHS. Clearly, $R_{l+1} \subseteq R_l$. First, we take care of some trivial cases.
If the constraints $c_{+p_i}$ is redundant (then $c_i$ must also be redundant) then it does not contribute to extreme points. Thus, the relevant case to consider is when $c_{+p_i}$ is not redundant. 

Now, we need to only consider those extreme points that lie on the constraint $c_{+p_i}$, as the other extreme points for $R_{l+1}$ will 
be same as for $R_{l}$ and hence integral. Consider the extreme point $p^*$ that lies on the constraint
$c_{+p_i}$. Wlog, let the variables in $c_{+p_i}$ be $p_1, \ldots, p_i$. We must have $\sum_1^i p_j = |L|$. Now, from the polytope $R_l$ remove all extreme points that have $\sum_1^{i} p_j = |L| + 1$, and consider the convex hull $R^-_l$ of remaining extreme points. Note that all extreme points of $R_l^-$ are in 
$R_{l+1}$, thus $R_l^- \subseteq R_{l+1}$. There could be two cases: one if there are no extreme points with $\sum_1^{i} p_j = |L| + 1$ in $R_l$, then $R_l = R^-_l \subseteq R_{l+1}$, and since $R_{l+1} \subseteq R_l$ we get $R_l = R_{l+1}$. Thus, all extreme points of $R_{l+1}$ will be $0/1$. 

The second (more interesting case) is when there are extreme points in $R_l$ with $\sum_1^{i} p_j = |L| + 1$ that 
get removed in $R^-_l$. Since $p^*$ is a point in $R_l$ it can be written as convex combination of extreme points. If all these extreme points satisfy  $\sum_1^{i} p_j \leq |L|$, then $p^*$ can be written as convex combination of extreme points of $R^-_l$ and since $R_l^- \subseteq R_{l+1}$, $p^*$ cannot be a extreme point. Or else, if one of the extreme points (in the convex combination) $P$ has $\sum_1^{i} p_j = |L| + 1$ with $p_s = 1$ ($1 \leq s \leq i$) then we can write $P$ as sum of two points within $R_l^-$: one which is same as $P$ except $p_s=0$, and other in which only $p_s = 1$ and other
components are zero. It is easy to check that these points are in $R_l^-$, since the inequalities allow for 
reducing any $p_j$ and clearly  $\sum_1^{i} p_j \leq |L|$. In such a manner, we can ensure
that all extreme points involved in the convex combination have $\sum_1^{i} p_j \leq |L|$, which reduces to the previous case.

\textbf{Extreme points are pure strategies}

Next, it is easy to see that for given $p_i$'s and $p_i^{*}$'s, with corresponding feasible $p_i^j$'s and $p_i^{*j}$'s, any convex combination $p^{@}_{i}$'s of $p_i$'s and $p^*_i$'s has a 
feasible solution $p_i^{@j}$'s which is the convex combination of $p_i^j$'s and 
$p_i^{*j}$'s. Since, any point in a convex set can be written as the convex combination
of its extreme points~\cite{}, it is enough to show the existence of feasible $p_i^j$'s 
for the extreme points in order to prove existence of feasible $p_i^j$'s for any point
in the convex polytope under consideration.

The extreme points of the given convex polytope has ones in $k' \leq k$ positions and 
all other zeros. The $k' \leq k$ arises due to the constraint $p_1 + \ldots + p_n\leq k$.
Consider the undirected bipartite graph linking the inspections node to the target nodes, 
with a link indicating that the inspection can audit the linked target. This graph is known
from our knowledge of $R$, and each link in the graph can be labeled by one of the $p_i^j$
variables. Let $S'$ be the set of targets picked by the ones in any
extreme points. We claim that there is a perfect matching from $S'$ to the the set of 
inspection resources (which we prove in next paragraph). Given such a perfect matching, assigning
$p_i^j = 1$ for every edge in the matching yields a feasible solution, which completes
the proof.

We prove the claim about perfect matching in the last paragraph. We do so by
contradiction. Assume there is no perfect matching, then there must be a set 
$S'' \subseteq S'$, such that $|N(S'')| < |S''|$ ($N$ is neighbors function, this statement holds
by the well known Hall's theorem). As $S ''\subseteq S'$ it must hold that 
$p_i = 1$ for all $i \in index(S'')$ (function index gives the indices of the set of targets).. Also, the set of targets $S''$ is audited
only by inspection resources in $N(S'')$ and, by definition of $C$ we must have a constraint (the function index gives the set of indices of the input set of targets)
$$
\sum_{i \in index(S'')} p_i \leq |N(S'')| \ .
$$
Using, $|N(S'')| < |S''|$, we get 
$$
\sum_{i \in index(S'')} p_i < |S''| \ .
$$
But, since $|index(S'')| = |S''|$, we conclude that all $p_i$ for targets in $S''$ cannot be 
one, which is a contradiction.
\end{proof}

\begin{proof}[Proof of Lemma~\ref{samealg}] 
Assume constraint $c$ is in the output of the naive algorithm. Restrict the constraint $c$ 
to be a non-implied constraint, i.e., if $c$ is given by $LHS(C) \leq RHS(c)$ then it  cannot be written as $\sum_{i=1}^{n} LHS(c_i) \leq \sum_{i=1}^{n} RHS(c_i)$ for any $c_1, \ldots, c_n$. Note that such a restriction does not change the region defined by
the set $C$. By the description of the 
naive algorithm, this constraint must correspond to a set of inspection resources, say $S$, and
the set of targets $T_S$ inspected only by inspection resources in $S$, and there exists
no subset of $T_S$ and $S$ such that the subset of targets is inspected only by the subset of inspection resources. The constraint $c$ is of the form $P(T_S) \leq |S|$, where $P(T_S)$ is the 
sum of the probability variables for targets in $T_S$. Now, for the intersection graph
representation, every node $x$ represents targets that are inspected by a given subset $x_s$ of inspection resources. For our case, we claim that every $t \in T_S$ is either equivalent to or linked to another target in $T_S$, otherwise 
we have the subset $\{t\} \subset S$ inspected only by $t_s \subset T_S$. Thus, the 
targets in $T_S$ form a connected induced sub-graph. Thus, we conclude that the 
$\mathsf{CONSTRAINT\_FIND}$ algorithm will consider this set of targets and find the
non-implied constraint $c$.

Next, assume that $\mathsf{CONSTRAINT\_FIND}$ finds a constraint $c$. As this
corresponds to a connected induced sub-graph (with targets as nodes), we obtain a set of targets 
$T$ and a set of inspection resources $\cup_{t \in T} F(t)$ such that targets in $T$ are audited
by $\cup_{t \in T} F(t)$ only. Thus, by definition of the construction of $c$, this constraint
is same as the constraint $c'$ obtained by the naive algorithm when it considers the set
$\cup_{t \in T} F(t)$. That is, $c$ will also be found by the naive algorithm.
\end{proof}

\begin{proof}[Proof of Lemma~\ref{graphlemma}] 
The proof below lists sufficient conditions under which the number of induced connected sub-graphs is polynomial is size. The proofs are constructive also, allowing extracting these sub-graphs in polynomial time.
\begin{itemize}
\item Graphs with $O(\log n)$ nodes. The maximum number of connected induced subgraphs in a 
graph with $t$ nodes is $2^t$ (take any subset of vertexes). Thus, clearly a graph with
$O(\log n)$ nodes will have polynomially many connected sub-graphs.
\item Graphs with constant max degree and constant number of nodes with degree $\geq 3$. The number of connected induced subgraphs in a 
tree with max degree $d$ and $t$ vertexes with degree $\geq 3$ is bounded from above 
by $2^{(2(d+1))^{t+1}} n^{(d+1)^{t+1}}$. To prove this result, denote by $T(n,d,t)$ the worst case number of connected 
induced sub-graphs in a graph with $n$ vertices, and max degree $d$ and $t$ vertices with degree $\geq 3$. 

Remove a vertex $X$ with degree $\geq3$ to get $k \leq d$ disconnected 
components. Each connected sub-graph in any component that was linked to $X$ could be combined with any connected sub-graph of any other component linked to $X$, yielding a new connected sub-graph. 
%Note that for any 
%connected sub-graph that was linked to $X$ in any component, there is another 
%connected sub-graph that was not linked to $X$ in that component (remove the vertex 
%that linked to $X$ from the sub-graph). Thus, the number of connected sub-graph that 
%was linked to $X$ in any component can be upper bounded by $T(n-1,d,t-1)/2$.
Thus, considering every subset of the $k$ components, we get
$$
T(n,d,t) \leq 2^k \left(T(n-1, d, t-1)\right)^k 
$$
Observing that $k < d+1$, and $T(n-1, d, t-1) \leq T(n, d, t-1)$ we get
$$
T(n,d,t) \leq \left( 2T(n, d, t-1) \right)^{d+1}
$$
Thus, we see that $T(n,d,t) = 2^{(2(d + 1))^{t+1}}(n)^{(d + 1)^{t+1}}$ satisfies the above equation. 

As part of the induction, the base case requires reasoning about a graph with max degree either $1$ or $2$
($t = 0, d = 1$ or $2$). Thus, we need to show that the connected sub-graphs is less than $n^{d+1}$, which is $n^2$ for max-degree $1$ and $n^3$ for max-degree $2$.
The max-degree $1$ case is trivial. For the max-degree $2$ case, such graphs can be decomposed efficiently into paths and cycles, and the number of connected sub-graphs
on cycles and paths is less that $n^2$.
%\item Bound on number of nodes with degree $\geq 3$.
\end{itemize}
\end{proof}

\begin{proof}[Proof of Lemma~\ref{targetlemma}]
First, assume $p_i^o$'s and $x_i^o$'s are an optimal point of $PX_n$.
If $x_n^o > 0$ then reducing the value of 
$x_n^o$ always yields a feasible point,
as the quadratic inequality is still satisfied. Also,
clearly reducing $x_n^o$ increases the objective value. Thus, we must have $x_n^o = 0$.

Next, reducing the value of 
$p_n^o$ by $\epsilon_p$ ($\leq \epsilon$) always yields a feasible point,
as the quadratic inequality is still satisfied and so are the linear inequalities. The values
of $\epsilon_p$ are chosen so that the new values of $p_n$ 
lie on our discrete grid for $p_n$. Thus, the new 
feasible point $F$  is $p_i^o$'s for $i = 1$ to $n-1$, $x_i^o$' for $i = 1$ to $n-1$ and $p_n = p_n^o - \epsilon_p$, $x_n = 0$. The objective at this feasible point $F$ is off from the optimal 
value by a linear combination of $\epsilon_p$ with constant coefficients, which is less than a constant times $\epsilon$. Then, the 
SOCP with the new values of $p_n$ yields an objective value at least as high 
as the feasible point $F$ on the grid. Thus, using our approach, we obtain a solution that differs from 
the optimum only by $\Theta(\epsilon)$. 
\end{proof}

\begin{proof}[Proof of Lemma~\ref{optimalprop}] 
observe that the quadratic inequality can
be rewritten as $$
p_n (x + \Delta_n) + \delta_{i,n} \leq p_i (x + \Delta_i) \ .
$$ 
Suppose  is an optimal point, and suppose for some $j$ we have the
strict inequality
$$
p_n^* (x^* + \Delta_n) + \delta_{j,n} < p_j^* (x^* + \Delta_j) \ .
$$
Let $p_j' = \min(0, \frac{p_n^* (x^* + \Delta_n) + \delta_{j,n}}{x^* + \Delta_j})$. Then, 
we claim that $p_1^*, \ldots, p^*_{j-1}, p_j', p^*_{j+1}, p^*_n, x^*$ is also an optimal point.
This is easy to see since decreasing $p_j$ is not restricted by any inequality other than
the quadratic inequality and the objective only depends on $p_n$. As a result, we 
can restrict the problem to be equalities for all those quadratic constraints for which
$p_n^* (x^* + \Delta_n) + \delta_{j,n} \geq 0$, and restrict $p_j = 0 $ in case $p_n^* (x^* + \Delta_n) + \delta_{j,n} < 0$.
\end{proof}

\renewcommand{\arraystretch}{1.7}
\begin{table}[t]
\begin{center}
\begin{tabular}{ | c | c | c | c | }
  \hline
  $\quad \mathbf{U_{D}^a} \quad$ & $\quad \mathbf{U_{D}^u} \quad$ & $\quad \mathbf{U_{A}^a} \quad$ & $\quad \mathbf{U_{A}^u} \quad$  \\
  \hline
  $0.614$  & $0.598$  & $0.202$ &  $0.287$ \\
  \hline
  $0.719$  & $0.036$  & $0.869$ &  $0.999$ \\
\hline
  $0.664$  & $0.063$  & $0.597$ & $0.946$ \\
\hline
  $0.440$  & $0.322$  & $0.023$ & $0.624$ \\
\hline
  $0.154$  & $0.098$  & $0.899$ & $0.902$ \\
\hline
  $0.507$  & $0.170$  & $0.452$ &  $0.629$ \\
\hline
  $0.662$  & $0.371$  & $1.000$ & $0.999$ \\
\hline
\end{tabular}\caption{Utility values for the counterexample}\label{counterexampledata}
\end{center}
\end{table}
\renewcommand{\arraystretch}{1.0}

\begin{figure}[t]
\begin{center}
\includegraphics[scale=0.50]{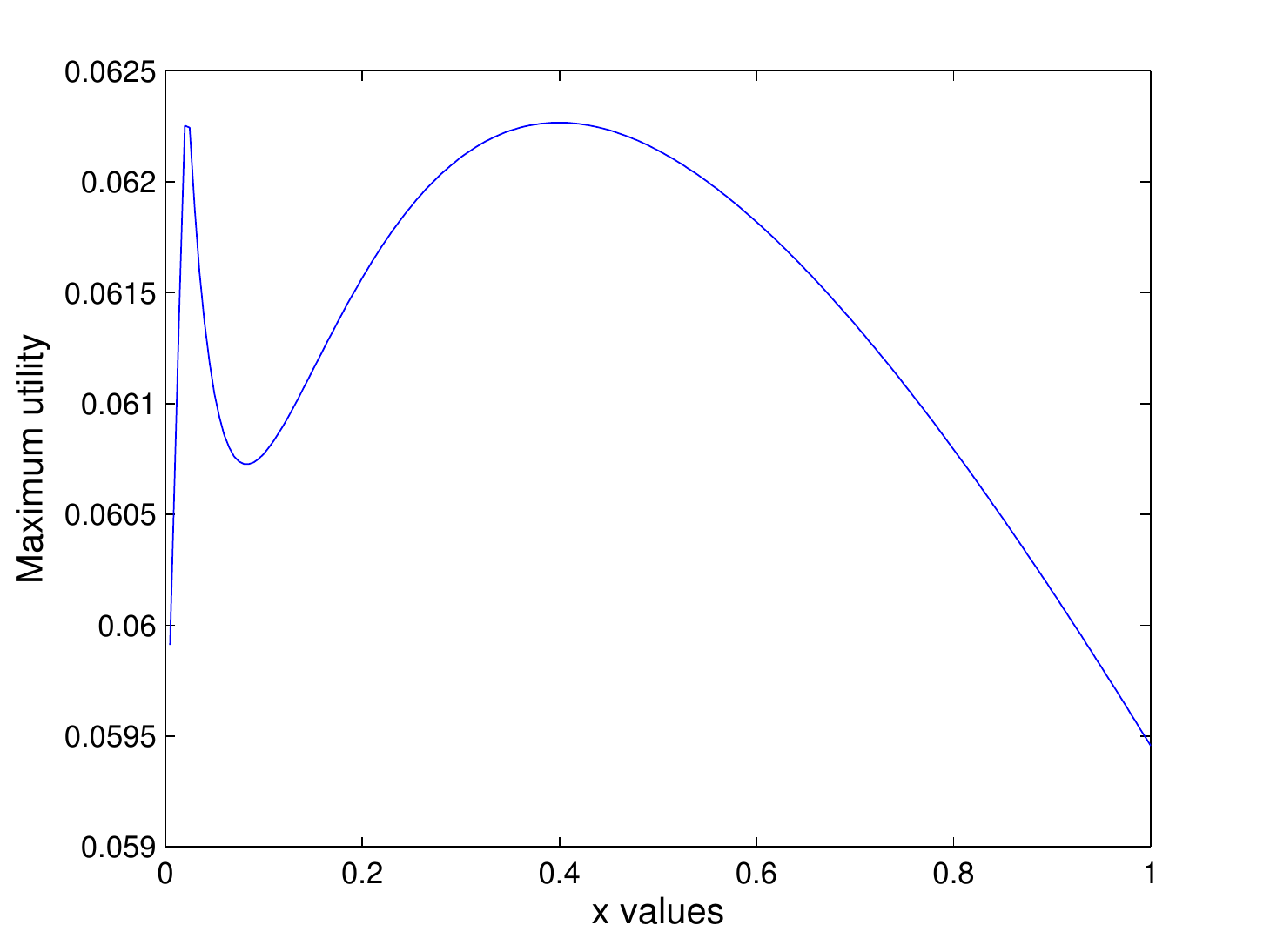}
\end{center}
\caption{Variation of maximum utility with $x$ showing multiple peaks.} \label{counterplot}
\end{figure}

\begin{proof}[Proof of Lemma~\ref{fixedxpn}] 
The fixed $x$ case is obvious, since the problem 
is a linear program for fixed $x$. The fixed $p_n$ case is equivalent to minimizing
$x$ with constraints that are of the form $f(x) \leq0$, where $f$ is a polynomial in $x$. 
For any polynomial constraint $f(x) \leq0$, it is possible to approximate the roots of the polynomial with an additive approximation factor of $2^{-l}$~\cite{schonhage1982fundamental} and hence 
find the intervals of $x$ that satisfies the constraint $f(x) \leq0$ within
additive approximation factor of $2^{-l}$. Doing so for the polynomially many constraints and finding the intersection of intervals yields the minimum value of $x$ with 
additive approximation factor of $2^{-l}$ that
satisfies all constraints. 
\end{proof}

\begin{proof}[Proof of Lemma~\ref{boundary}] 
We can derive an 
easy contradiction for this scenario. Assume $p_n^*, x$ is optimal.
As $f_b$ is continuous in $x$, if all inequalities are strict, then it is possible to 
increase the value of optimal $p_n^*$ by a small amount, without violating the 
constraints. Clearly, this increased value of $p_n$ results in a higher objective value,
contradicting the assumption that $p_n^*, x^*$ is optimal.

\end{proof}

\section{Maximum Value of Objective is Not Single-peaked} \label{counterexample}

As stated above, the solution of the optimization problem is 
not single peaked in punishment $x$. Here we show the counterexample that proves
this fact. We choose a problem instance with just one defender resource and
7 targets, and we consider only the case when the seventh target is the target under attack. The value of $a$ was chosen to be $0.01$, and $x$ was discretized with 
interval size of $0.005$. The various values of utilities are shown in Table~\ref{counterexampledata}. The variation of maximum utility with $x$ is shown
in Figure~\ref{counterplot}. It can be seen that the maximum value is not single 
peaked in $x$.

\section{Model enhancement.} We state an extension to the model that captures immediate losses that the defender suffers by imposing a punishment, for example, firing or suspending an employee requires time and effort to find a replacement. Mathematically, we can account for such loss by including an additional term in the objective of our optimization:
$$
p_{n}(\Delta_{D,n} - a_1 x) - a x
$$
where $-a_1 x$ captures the loss from imposing punishment. It is not hard to check that we still get a FPT with the above objective. The reason for that is as because for any discrete change $\epsilon$ in $x$,
the term $- p_n a_1 x$ changes by at max $\theta(\epsilon)$ amount given constant bit precision.

We also still get the same FPTAS result by a minor modification to Lemma~\ref{optimalprop}. In that lemma, we need to consider the case of $\Delta_{D,n} - a_1 x < 0$ separately, and it is not too hard to see that in this case the optimal $p_n$ is $0$. We can solve this case of
$p_n = 0$ using Lemma~\ref{fixedxpn}.

For target-specific punishment the objective would change to
$$
 p_{n} (\Delta_{D,n} - a_n' x_n) - \sum_{j \in \{1, \ldots, n\}} a_j x_j
$$
Again, using the exact same argument as in Lemma~\ref{targetlemma}, it is not hard to see that at the optimal point $x_n$ should be $0$. Once $x_n$ is $0$, the objective
is linear in the variables and we can apply the same SOCP reduction as earlier.

\end{document}